\documentclass[letterpaper, 10pt, conference]{ieeeconf}  

\IEEEoverridecommandlockouts                              
\usepackage{mathtools} 

\usepackage[sort,compress]{cite}
\usepackage{amsfonts,dsfont,amssymb,bm}

\usepackage[amsmath,thmmarks]{ntheorem}
\usepackage{subcaption}

\usepackage{microtype}

\usepackage{algorithm}
\usepackage{algpseudocode}
\usepackage{xcolor}

\newcommand\tildeO{\tilde{\mathcal O}}
\newcommand\reals{\mathds{R}}
\newcommand\integers{\mathds{N}}
\newcommand\EXP{\mathds{E}}
\newcommand\PR{\mathds{P}}
\newcommand\IND{\mathds{1}}
\DeclareMathOperator\TR{Tr}
\DeclareMathOperator\norm{norm}
\DeclareMathOperator\ball{ball}

\newcommand*\TRANS{{\mathpalette\doTRANS\empty}}
\makeatletter
\newcommand*\doTRANS[2]{\raisebox{\depth}{$\m@th#1\intercal$}}
\makeatother

\usepackage[amsmath,thmmarks]{ntheorem}

\newtheorem{assumption}{Assumption}

\begin{document}

\title{A modified Thompson sampling-based learning algorithm for unknown linear systems}

\author{Mukul Gagrani, Sagar Sudhakara, Aditya Mahajan, Ashutosh Nayyar, and Yi Ouyang%
  \thanks{Mukul Gagrani is with Qualcomm AI research, San Diego. (email: mgagrani@qti.qualcomm.com)}%
  \thanks{Sagar Sudhakara and Ashutosh Nayyar are with the Department of Electrical and Computer Engineering, University of Southern California, Los Angeles, CA, USA. (email: sagarsud@usc.edu, ashutosn@usc.edu)}%
  \thanks{Aditya Mahajan is with the department of Electrical and Computer Engineering, McGill University, Montreal, QC, Canada. (email: aditya.mahajan@mcgill.ca)}%
  \thanks{Yi Ouyang is with Preferred Networks America, Burlingame, CA, USA (email: ouyangyi@preferred-america.com)}%
  \thanks{The work at USC was supported by NSF grants ECCS 1750041 and ECCS 2025732 and the work of Aditya Mahajan was supported in part by  the Innovation for Defence Excellence and Security (IDEaS) Program through grant CFPMN2-30.}%
}

\maketitle

\begin{abstract}
We revisit the Thompson sampling-based learning algorithm for controlling an unknown linear system with quadratic cost proposed in \cite{ouyang2019posterior}. This algorithm operates in episodes of dynamic length and it is shown to have a regret bound of $\tildeO(\sqrt{T})$, where $T$ is the time-horizon. The regret bound of this algorithm is obtained under a technical assumption on the induced norm of the closed loop system. We propose a variation of this algorithm that enforces a lower bound $T_{\min}$ on the episode length. We show that a careful choice of $T_{\min}$ (that depends on the uncertainty about the system model) allows us to recover the  $\tildeO(\sqrt{T})$ regret bound under a milder technical condition about the closed loop system.
\end{abstract}

\section{Introduction}

The problem of learning an optimal policy for a system with linear dynamics and quadratic cost with unknown parameters has received considerable attention in the literature. Historically, the focus has been on developing algorithms which asymptotically learn optimal policies using techniques from adaptive control and reinforcement learning~\cite{astrom1994adaptive,Caines:1988,bradtke1993reinforcement, bradtke1994adaptive}. In recent years, the emphasis has shifted towards developing algorithms with finite-time regret guarantees. 

Broadly speaking, three classes of learning algorithms have been considered in the literature: optimism in the face of uncertainty~\cite{campi1998adaptive,abbasi2011regret,cohen2019learning,abeille2020efficient}, certainty equivalence~\cite{dean2018regret,mania2019certainty,faradonbeh2020input,simchowitz2020naive,faradonbeh2020adaptive}, and Thompson sampling~\cite{ouyang2017control,abeille2018improved,faradonbeh2020adaptive,gagrani2021thompson}. These algorithms provide two kinds of regret guarantees: frequentist and Bayesian. In the frequentist setting, it is established that the regret for the unknown system is bounded with high probability (with respect to the distribution of the process noise and the randomness introduced by the algorithm). In the Bayesian setting, it is assumed that there is a prior on the unknown system parameters and it is established that the expected regret is bounded (where the expectation is with respect to the prior, the distribution of the process noise, and the randomness introduced by the algorithm). These two notions of regret are different and, since the per-step cost is not bounded, one form of the regret does not imply the other. 

In this paper, we revisit a recently proposed algorithm for establishing Bayesian regret called Thompson sampling with dynamic episodes (TSDE)~\cite{ouyang2019posterior}. 
The main result of~\cite{ouyang2019posterior} is to show that the Bayesian regret of TSDE accumulated up to time~$T$ is bounded by $\tildeO(\sqrt{T})$, where the $\tildeO(\cdot)$ notation hides constants and poly-logarithmic factors. This result was derived under a technical assumption on the induced norm of the closed loop system. In this paper, we present a variation of the TSDE algorithm and obtain a $\tildeO(\sqrt{T})$ bound on the Bayesian regret by imposing a much milder technical assumption.

\section{Model and problem formulation}

We consider the same model as~\cite{ouyang2019posterior}. For the sake of completeness, we present the model below.

Consider a linear system with state $x_t \in \reals^n$, control input $u_t \in \reals^m$, and disturbance $w_t \in \reals^n$. For the ease of exposition, we assume that the system starts from an initial state $x_1 = 0$. The state evolves over time according to 
\begin{equation}
  x_{t+1} = A x_t + B u_t + w_t, \quad t \ge 1, 
\end{equation}
where $A \in \reals^{n \times n}$ and $B \in \reals^{n \times m }$ are the system dynamics matrices. The noise $\{w_t\}_{t \ge 1}$ is an independent and identically distributed Gaussian process with $w_t \sim \mathcal{N}(0, \sigma_w^2 I)$.

\begin{remark}
  In~\cite{ouyang2019posterior}, it was assumed that $\sigma_w^2 = 1$. Using a general $\sigma_w^2 > 0$ does not fundamentally change any of the results or the proof arguments.
\end{remark}

At each time~$t$, the system incurs a per-step cost given by
\begin{equation}
  c(x_t,u_t) = x_t^\TRANS Q x_t + u_t^\TRANS R u_t,
\end{equation}
where $Q$ and $R$ are positive definite matrices. 

Let $\theta^\TRANS = [A, B]$ denote the parameters of the system.  $\theta \in \reals^{d \times n}$, where $d = n+m$. The performance of any policy $\pi = (\pi_1, \pi_2, \dots)$ is measured by the long-term average cost given by 
\begin{equation}
  J(\pi;\theta) = \limsup_{T \to \infty} \frac{1}{T}
  \EXP^\pi \Bigl[ \sum_{t=1}^T c(x_t, u_t) \Bigr].
\end{equation}
Let $J(\theta)$ denote the minimum of $J(\pi;\theta)$ over all policies. It is well known~\cite{Astrom1970} that if the pair $(A,B)$ is stabilizable, then $J(\theta)$ is given by
\[
  J(\theta) = \sigma_w^2 \TR(S(\theta)),
\]
where $S(\theta)$ is the unique positive semi-definite solution of the following Riccati equation:
\begin{multline}
  S(\theta) = Q + A^\TRANS S(\theta) A \\
  - A^\TRANS S(\theta) B( R + B^\TRANS S(\theta) B)^{-1} B^\TRANS S(\theta)A. \label{eq:ric_eq}
\end{multline}
Furthermore, the optimal control policy is given by
\begin{equation}
  u_t = G(\theta) x_t,
\end{equation}
where the gain matrix $G(\theta)$ is given by 
\begin{equation}\label{eq:ric_gain}
  G(\theta) = -(R + B^\TRANS S(\theta) B)^{-1} B^\TRANS S(\theta) A.
\end{equation}

As in~\cite{ouyang2019posterior}, we are interested in the setting where the system parameters  are unknown. We denote the unknown parameters by a random variable $\theta_1$ and assume that  there is a prior distribution on $\theta_1$. The Bayesian regret of a policy $\pi$ operating for horizon $T$ is defined by 
\begin{equation}
  R(T;\pi) = \EXP^\pi \Bigl[ \sum_{t=1}^T c(x_t, u_t) - T J(\theta_1) \Bigr],
\end{equation}
where the expectation is with respect to the prior on $\theta_1$, the noise processes, the initial conditions, and the potential randomizations done by the policy~$\pi$.

\section{Thomson sampling based learning algorithm}

As in~\cite{ouyang2019posterior}, we assume that the unknown model parameters $\theta$ lie in a compact subset $\Omega_1$ of $\reals^{d \times n}$. We use $p|\Omega$ to denote the restriction
of probability distribution $p$ on the set $\Omega$.
We assume that there is a prior $\mu_1$ on $\Omega_1$ which satisfies the following assumption.
\begin{assumption}\label{ass:prior}
  There exist $\hat \theta_1(i) \in \reals^{d}$ for $i \in \{1, \dots, n\}$ and a positive definite matrix $\Sigma_1 \in \reals^{d\times d}$ such that for any $\theta \in \reals^{d\times n}$, $\mu_1 = \bar \mu_1 \bigr|_{\Omega_1}$, where 
  \[
    \bar \mu_1(\theta) = \prod_{i=1}^n \bar \mu_1(\theta(i))
    \quad\text{and}\quad
    \bar \mu_1(\theta(i)) = \mathcal{N}(\hat \theta_1(i), \Sigma_1).
  \]
\end{assumption}

We maintain a posterior distribution $\mu_t$ on $\Omega_1$ based on the history $(x_{1:t-1}, u_{1:t-1})$ of the observations until time~$t$. From standard results in linear Gaussian regression~\cite{sternby1977consistency}, we know that the posterior is a truncated Gaussian distribution
\[
  \mu_t(\theta) = \biggl[ 
    \prod_{i=1}^n \bar \mu_t( \theta(i) ) 
  \biggr] \biggr|_{\Omega_1}
\]
where $\bar \mu_t(\theta(i)) = \mathcal{N}(\hat \theta_t(i), \Sigma_t)$ and $\{ \hat \theta_t(i) \}_{i=1}^n$ and $\Sigma_t$ can be updated recursively as follows:
\begin{align}
  \hat \theta_{t+1}(i) &= \hat \theta_t(i) + 
  \frac{\Sigma_t z_t ( x_{t+1}(i) - \hat \theta_t(i)^\TRANS z_t)}
  {\sigma_w^2 + z_t^\TRANS \Sigma_t z_t},
  \label{eq:theta}
  \\
  \Sigma_{t+1}^{-1} &= \Sigma_t^{-1} + \frac{1}{\sigma_w^2} z_t z_t^\TRANS,
  \label{eq:Sigma}
\end{align}
where $z_t = [x_t^\TRANS, u_t^\TRANS]^\TRANS$.

\subsection{Thompson sampling with dynamic episodes algorithm}

We now present a variation of the Thompson sampling with dynamic episodes (TSDE) algorithm of~\cite{ouyang2019posterior}. As the name suggests, the algorithm operates in episodes of dynamic length. The key difference from~\cite{ouyang2019posterior} is that we enforce that each episode is of a minimum length $T_{\min}$. The choice of $T_{\min}$ will be explained later. 

Let $t_k$ and $T_k$ denote the start time and the length of episode~$k$, respectively. Episode~$k$ has a minimum length of $T_{\min}$ and ends when the length of the episode is strictly larger than the length of the previous episode (i.e., $t - t_k > T_{k-1}$) or at the first time after $t_k + T_{\min}$ when the determinant of the covariance~$\Sigma_t$ falls below half of its value at time $t_k$, i.e., $\det \Sigma_t < \tfrac 12 \det \Sigma_{t_k}$. Thus, 
\begin{equation}\label{eq:stopping}
    t_{k+1} = \min \left\{ t > t_k + T_{\min} 
      \,\middle|\, \begin{lgathered}
        t - t_k > T_{k-1} \text{ or } \\
      \det \Sigma_t < \tfrac 12 \det \Sigma_{t_k} 
  \end{lgathered} \right\}.
\end{equation}

Note that  the stopping condition~\eqref{eq:stopping} implies that
\begin{equation}
  T_{\min} + 1
  \le 
  T_k 
  \le 
  T_{k-1} +1 
  , \quad \forall k
  \label{eq:min-length}
\end{equation}
If we select $T_{\min} = 0$ in the above algorithm, we recover the stopping condition of~\cite{ouyang2019posterior}. 

The TSDE algorithm works as follows. At the beginning of episode~$k$, a parameter $\bar \theta_k$ is sampled from the posterior distribution $\mu_{t_k}$. During the episode, the control inputs are generated using the sampled parameters $\bar \theta_k$, i.e., 
\begin{equation}
  u_t = G(\bar \theta_k) x_t, \quad
  t_k \le t < t_{k+1}.
\end{equation}
The complete algorithm is presented in Algorithm~\ref{alg:TSDE}. 

\begin{algorithm}[!t]
\caption{\texttt{TSDE}}
\label{alg:TSDE}
\begin{algorithmic}[1]
  \State \textbf{input:} $\Omega_1$, $\hat \theta_1$, $\Sigma_1$
  \State \textbf{initialization:} $t \gets 1$, $t_0 \gets -T_{\min}$, $T_{-1}
  \gets T_{\min}$, $k \gets 0$.
\For{$t = 1, 2, \dots $}
  \State observe $x_t$
  \State update $\bar \mu_t$ according to~\eqref{eq:theta}--\eqref{eq:Sigma}
  \State \textbf{if} $(t - t_k > T_{\min})$ and 
  \State \null \quad $\bigl( (t-t_k > T_{k-1}) \hbox{ or } (\det \Sigma_t <
  \tfrac12 \det \Sigma_{t_k}) \bigr)$ 
  \State \textbf{then} 
  \State \null \quad $T_k \gets t - t_k$, $k \gets k + 1$, $t_k \gets t$
  \State \null \quad sample $\bar \theta_k \sim \mu_t$
  \State \textbf{end if}
  \State Apply control $u_t = G(\bar \theta_k) x_t$
\EndFor
\end{algorithmic}
\end{algorithm}
  
\subsection{A technical assumption and the choice of minimum episode length}
For each $\theta^\TRANS = [A, B]$, we define a 4-dimensional row-vector $\eta(\theta)$ as follows:
\begin{equation}
    \eta(\theta) := (\|\theta\|, \TR(S(\theta)), \|S(\theta)\|, \|[I, G(\theta)^\TRANS]^\TRANS \|),
\end{equation}
where $S(\theta)$ is the solution of the Riccati equation in \eqref{eq:ric_eq}  and $G(\theta)$ is the optimal gain matrix  defined in \eqref{eq:ric_gain}. 
\begin{definition}\label{defn:type}
Let $M_{\theta}, M_J, M_S, M_G, \alpha, \delta$ be positive constants such that $\alpha \geq 1$ and $0 < \delta <1$. We say that the uncertainty set $\Omega_1$ is of \emph{Type} $(M_{\theta}, M_J, M_S, M_G, \alpha, \delta)$ if the following conditions hold:
\begin{enumerate}
    \item For all $\theta \in \Omega_1$, 
    \begin{align}\label{defn:defn1}
        \eta(\theta) \leq (M_{\theta}, M_J, M_S, M_G)
    \end{align}
where the inequality is component-wise.
    \item For any $\theta, \phi \in \Omega_1$ with $\theta^\TRANS =
  [A_\theta, B_\theta]$ and for any integer $t \ge 1$, 
  \[
    \| (A_\theta + B_\theta G(\phi))^t \| \le \alpha\delta^t.
  \]
\end{enumerate}
\end{definition}

\begin{assumption}\label{ass:type}
  We assume that the uncertainty set $\Omega_1$ is of \emph{Type} $(M_{\theta}, M_J, M_S, M_G, \alpha, \delta)$, where $M_{\theta}$, $M_J$, $M_S$, $M_G$, $\alpha$, $\delta$ are positive constants such that $\alpha \geq 1$ and $0 < \delta <1$.
\end{assumption}
The following simple observation plays a critical role in analyzing the regret of TSDE.
\begin{lemma}\label{lem:t-min}
   Suppose Assumption \ref{ass:type} is true. Define
   \begin{equation}\label{eq:t-min}
    T^*_{\min} =  
      \biggl\lceil \frac{ \log \alpha }{ - \log \delta } \biggr\rceil.
  \end{equation}
  Then, for 
  $\theta, \phi \in \Omega_1$ with $\theta^\TRANS = [A_\theta, B_\theta]$, we
  have 
  \begin{equation}\label{eq:bound}
    \| (A_\theta + B_\theta G(\phi))^{T^*_{\min} + 1} \| < 1
  \end{equation}
\end{lemma}
\begin{proof}
  The proof follows immediately from Assumption~\ref{ass:type} and the definition of $T^*_{\min}$.
\end{proof}

Before presenting our regret analysis under Assumption~\ref{ass:type}, we present two  special cases of this assumption. The first case is identical to the assumption made in \cite{ouyang2019posterior} about the uncertainty set.

\begin{assumption}\label{ass:norm}
  Let $M_{\theta}$, $M_J$, $M_S$, $M_G$, $\alpha$, $\delta$ be positive constants such that $\alpha \geq 1$ and $0 < \delta <1$.
  Assume that the uncertainty set $\Omega_1$  satisfies the following conditions:
\begin{enumerate}
    \item Equation \eqref{defn:defn1} holds for all $\theta \in \Omega_1$.
    \item For any $\theta, \phi \in \Omega_1$ with $\theta^\TRANS =
  [A_\theta, B_\theta]$,
  \[
    \| (A_\theta + B_\theta G(\phi)) \| \le \delta.
  \]
\end{enumerate}
\end{assumption}

In \cite{ouyang2019posterior}, part 2) of Assumption \ref{ass:norm} was stated explicitly. In addition, the uncertainty set was assumed to be compact, which ensures part 1) of Assumption \ref{ass:norm}.  

An uncertainty set that satisfies Assumption \ref{ass:norm} also satisfies Assumption \ref{ass:type} with $\alpha =1$. This is because if $\| (A_\theta + B_\theta G(\phi)) \| \le \delta$, then
\begin{align}
    \| (A_\theta + B_\theta G(\phi))^t \| \le \left(\| (A_\theta + B_\theta G(\phi)) \|\right)^t \le \delta^t.
\end{align}
Thus, Assumption \ref{ass:type} is weaker than Assumption \ref{ass:norm} used in \cite{ouyang2019posterior}.

For a square matrix $A$, let $\rho(A)$ denote the spectral radius of matrix $A$.The next assumption can also be viewed as a special case of Assumption \ref{ass:type}. 

\begin{assumption}\label{ass:radius}
  Let $M_{\theta}$, $M_J$, $M_S$, $M_G$, $\alpha$, $\tilde\delta$ be positive constants  such that $\alpha \geq 1$ and $0 < \tilde\delta <1$.
  Assume that the uncertainty set $\Omega_1$  satisfies the following conditions:
\begin{enumerate}
    \item Equation \eqref{defn:defn1} holds for all $\theta \in \Omega_1$.
    \item For any $\theta, \phi \in \Omega_1$ with $\theta^\TRANS =
  [A_\theta, B_\theta]$,
  \[
   \rho(A_\theta + B_\theta G(\phi)) \le \tilde\delta.
  \]
\end{enumerate}
\end{assumption}

We note that Assumption~\ref{ass:radius} is weaker than Assumption \ref{ass:norm} (since $\rho(A) \leq ||A||$ for any matrix $A$).
Consider, for example, a family of matrices $A_q = \begin{bmatrix}
 \tilde\delta & q \\
 0 & \tilde\delta
 \end{bmatrix}$, where $q \in \integers $ and $0 < \tilde\delta < 1$. For each $q$, the spectral radius of $A_q$ is $\tilde\delta$ while its norm is at least $q$. Thus, each $A_q$ satisfies Assumption~\ref{ass:radius} but not Assumption~\ref{ass:norm}.

The following lemma shows that an uncertainty set that satisfies Assumption \ref{ass:radius} also satisfies Assumption \ref{ass:type} for some constants $\alpha \ge 1$ and $\tilde\delta < \delta < 1$.

\begin{lemma}\label{lem:bound}
  Suppose Assumption \ref{ass:radius} is true. Then,  there exist $\alpha \geq 1$ and $\tilde\delta < \delta < 1$ such that for any $\theta, \phi \in \Omega_1$ with $\theta^\TRANS =   [A_\theta, B_\theta]$ and for any integer $t \ge 1$,  
  \[
    \| (A_\theta + B_\theta G(\phi))^t \| \le \alpha\delta^t.
  \]
  \end{lemma}

\begin{proof}
   Define $\varepsilon= \delta -\tilde\delta$.
  Let $\mathcal{L} = \{ A_{ \theta} + B_{ \theta} G( \phi) :  \theta,  \phi \in  \Omega_1 \}$. Since $\Omega_1$ is compact, so is $\mathcal{L}$. Now for any $L \in \mathcal{L}$, there exists a norm (call it $\norm_L$) such that $\norm_L(L) < \rho(L) + \varepsilon \le \tilde\delta + \varepsilon = \delta$. 

  Since norms are continuous, there is an open ball centered at $L$ (let's call this $\ball_L$) such that for any $H \in \ball_L$, we have $\norm_L(H) < \delta$. Consider the collection of open balls $\{ \ball_L : L \in \mathcal{L} \}$. This is an open cover of compact set $\mathcal{L}$. So, there is a finite sub-cover. Let's denote this sub-cover by $\ball_{L_1}, \dots, \ball_{L_{\ell}}$. By equivalence of norms, there is a finite constant $\alpha_k$ such that $\| A \| \le \alpha_{L_k} \norm_{L_k}(A)$ for any matrix $A$, for all $k \in \{1, \dots, \ell\}$. Let $\alpha = \max(1, \max_k \alpha_{L_k})$. 

  Now consider an arbitrary $H \in \mathcal{L}$. It belongs to $\ball_{L_k}$ for some $k \in \{1, \dots, \ell\}$. Therefore, $\norm_{L_k}(H) < \delta $. Hence, for any integer $t$, the above inequalities and the submulitplicity of norms give that $\| H^t \| \le \alpha_{L_k} \norm_{L_k}(H^t) \le \alpha (\norm_{L_k}(H))^t < \alpha(\delta)^t$. 
\end{proof}
\begin{remark}
  Condition~2 in Definition~\ref{defn:type} states that $A_\theta + B_\theta G(\phi)$ is \emph{uniformly} exponentially stable for all all $\theta, \phi \in \Omega_1$. Condition~2 of Assumption~\ref{ass:radius} states that $A_\theta + B_\theta G(\phi)$ is \emph{uniformly} asymptotically stable for all $\theta,\phi \in \Omega_1$. For linear systems, asymptotic stability implies exponential stability. In Lemma~\ref{lem:bound}, we are effectively showing that when the uncertainty set is compact, uniform asymptotic stability implies uniform exponential stability.
\end{remark}

\subsection{Regret bounds}

The following result provides an upper bound on the regret of the proposed algorithm.

\begin{theorem}\label{thm:regret}
Under Assumptions~\ref{ass:prior} and~\ref{ass:type} {and with $T_{\min} \geq T^*_{\min}$}, the regret of TSDE is upper bounded by 
  \begin{equation}
    R(T; {\tt TSDE}) \le \tildeO(\sigma_w^2 (n+m) \sqrt{nT}).
  \end{equation}
\end{theorem}
The proof is presented in the next section.
\begin{remark}
 The constants hidden in the $\tildeO(\cdot)$ notation in Theorem~\ref{thm:regret} depend only on the type  $(M_{\theta}, M_J, M_S, M_G, \alpha, \delta)$ of the uncertainty set $\Omega_1$. In particular, these constants do not depend on $n$, $m$, and $T$.
\end{remark}
\section{Regret analysis}

For the ease of notation, we use $R(T)$ instead of $R(T; \texttt{TSDE})$ in this section. Let $K_T$ denote the number of episodes until horizon~$T$. 
Following the exact same steps as~\cite{ouyang2019posterior}, we can show that
\begin{equation}\label{eq:regret-decompose}
  R(T) = R_0(T) + R_1(T) + R_2(T)
\end{equation}
where
\begin{align}
  R_0(T) &= \EXP\bigg[ \sum_{k=1}^{K_T} T_k J(\bar \theta_k) \biggr] - T \EXP[ J(\theta_1) ],
  \\
  R_1(T) &= \EXP\Bigg[ \sum_{k=1}^{K_T} \sum_{t = t_k}^{t_{k+1} - 1}
    \bigl[
      x_t^\TRANS S(\bar \theta_k) x_t - x_{t+1}^\TRANS S(\bar \theta_k) x_{t+1} 
  \bigr] \Biggr]
  \\
  R_2(T) &=  \EXP\Bigg[ \sum_{k=1}^{K_T} \sum_{t = t_k}^{t_{k+1} - 1} 
    \bigl[ (\theta_1^\TRANS z_t)^\TRANS S(\bar \theta_k) \theta_1^\TRANS z_t
      \notag \\[-12pt]
  & \hskip 10em -
  (\theta_k^\TRANS z_t)^\TRANS S(\bar \theta_k) \theta_k^\TRANS z_t
\bigr] \Biggr]
\end{align}
We establish the bound on $R(T)$ by individually bounding $R_0(T)$, $R_1(T)$,
and $R_2(T)$. 

\begin{lemma}\label{lem:regret}
  The terms in~\eqref{eq:regret-decompose} are bounded as follows:
  \begin{enumerate}
    \item $R_0(T) \le \tildeO(\sigma_w^2 \sqrt{(n+m) T})$.
      \vskip 0.5\baselineskip
    \item $R_1(T) \le \tildeO(\sigma_w^2 \sqrt{(n+m) T})$.
      \vskip 0.5\baselineskip
    \item $R_2(T) \le \tildeO( \sigma_w^2 (n+m) \sqrt{n T})$.
  \end{enumerate}
\end{lemma}
Combining Lemma~\ref{lem:regret} with equation~\eqref{eq:regret-decompose} establishes Theorem~\ref{thm:regret}.
Before presenting the proof of Lemma~\ref{lem:regret}, we establish some preliminary results.

\subsection{Preliminary results}

Let $X_T = \sigma_w+\max_{1 \le t \le T} \| x_t\|$ denote the maximum of the norm of
the state plus the noise standard deviation. 

\begin{lemma}\label{lem:Xq}
  For any $q \ge 1$ and any $T \ge 1$, 
  \[
    \EXP\Big[ \frac{X_T^q}{\sigma_w^q}\Big] \le \mathcal{O} ({(\log T)}^{q/2}).
  \]
\end{lemma}
See Appendix~\ref{app:Xq} for proof.

\begin{lemma} \label{lem:log-bound}
  For any $q \ge 1$, we have 
  \[
    \EXP\Big[ \frac{X_T^q}{\sigma_w^q} \log
    \Big(\frac{X_T^2}{\sigma_w^2}\Big) \Big] \le  \tildeO(1).
  \]
\end{lemma}
See Appendix~\ref{app:log-bound} for proof.

\begin{lemma}\label{lem:Kt}
  The number of episodes is bounded by 
  \[
    K_T \le \mathcal{O}\Bigl(
    \sqrt{(n+m) T \log\left(T \frac{X_T^2}{\sigma_w^2}\right)} \Bigr).
  \]
\end{lemma}
See Appendix~\ref{app:Kt} for proof.

\begin{remark}
  The statement of Lemmas~\ref{lem:Xq} and~\ref{lem:Kt} are the same as that
  of the corresponding lemmas in~\cite{ouyang2019posterior}. The proof of
  Lemma~\ref{lem:Xq} in~\cite{ouyang2019posterior} relied on
  Assumption~\ref{ass:norm}. Since we impose a
  weaker assumption, our proof is more involved. The proof of
  Lemma~\ref{lem:Kt} is similar to the proof of \cite[Lemma~3]{ouyang2019posterior}.
  However, since our TSDE algorithm is different from that in~\cite{ouyang2019posterior},
  some of the details of the proof are different. 
\end{remark}

\subsection{Proof of Lemma~\ref{lem:regret}}\label{app:regret}

We now prove each part of Lemma~\ref{lem:regret} separately.

\subsubsection{Proof of bound on $R_0(T)$}
Following exactly the same argument as the proof
of~\cite[Lemma~5]{ouyang2019posterior}, we can show that
\begin{equation}
  R_0(T) \le \mathcal{O}(\sigma_w^2 \EXP[K_T]).
\end{equation}
Substituting the result of Lemma~\ref{lem:Kt}, we get
\begin{align*}
  R_0(T) &\le \mathcal{O}\Bigl(\sigma_w^2  \EXP\Bigl[ 
  \sqrt{(n+m)T \log(T X_T^2 / \sigma_w^2)} \Bigr] \Bigr)
  \\
  &\stackrel{(a)}\le \mathcal{O}\Bigl(\sigma_w^2  \sqrt{(n+m) T \log(T \EXP[ X_T^2 / \sigma_w^2
  ]}) \Bigr)\\
  &\stackrel{(b)}\le \tildeO\bigl(\sigma_w^2  \sqrt{(n+m) T} \bigr)
\end{align*}
where $(a)$ follows from  Jensen's inequality and $(b)$ follows from
Lemma~\ref{lem:Xq}.

\subsubsection{Proof of bound on $R_1(T)$}
Following exactly the same argument as in the proof
of~\cite[Lemma~6]{ouyang2019posterior}, we can show that
\begin{equation}\label{eq:R1-interim}
  R_1(T) \le \mathcal{O}( \EXP[ K_T X_T^2 ] )
\end{equation}
Substituting the result of Lemma~\ref{lem:Kt}, we get
\begin{equation}\label{eq:R1}
  R_1(T) \le \mathcal{O}\Bigl(
    \sqrt{(n+m)T}\, \EXP\bigl[ X_T^2 \sqrt{\log(T X_T^2 / \sigma_w^2)} \bigr]
  \Bigr)
\end{equation}
Now, consider the term
\begin{align}
  \EXP\bigl[ X_T^2 \sqrt{\log(T X_T^2 / \sigma_w^2)} \bigr] &\stackrel{(a)}\le
  \sqrt{ \EXP\bigl[ X_T^4 \bigr] \EXP\bigl[ \log(T X_T^2/ \sigma_w^2) \bigr]}
  \notag\\
  &\stackrel{(b)}\le
  \sqrt{ \EXP\bigl[ X_T^4 \bigr] \log(T \EXP[ X_T^2/ \sigma_w^2] ) }
  \notag\\
  &\stackrel{(c)}\le
  \tildeO(\sigma_w^2)
  \label{eq:R1:1}
\end{align}
where $(a)$ follows from Cauchy-Schwartz inequality, $(b)$ follows from
Jensen's inequality, and $(c)$ follows from Lemma~\ref{lem:Xq}. 

Substituting~\eqref{eq:R1:1} in~\eqref{eq:R1}, we get the bound on $R_1(T)$.

\subsubsection{Proof of bound on $R_2(T)$}

As in~\cite{ouyang2019posterior}, we can bound the inner summand in $R_2(T)$
as
\[
  \| S(\bar \theta_k)^{0.5} \theta_1^\TRANS z_t \|^2 - 
  \| S(\bar \theta_k)^{0.5} \theta_k^\TRANS z_t \|^2 
  \le 
  \mathcal{O}( X_T \| (\theta_1 - \bar \theta_k)^\TRANS z_t \|).
\]
Therefore, 
\[
  R_2(T) \le \mathcal{O}\biggl(
    \EXP\Bigg[ X_T \sum_{k=1}^{K_T} \sum_{t = t_k}^{t_{k+1} - 1} 
  \| (\theta_1 - \bar \theta_k)^\TRANS z_t \| \biggr] \biggr),
\]
which is same as~\cite[Eq.~(45)]{ouyang2019posterior}. Now, by simplifying
the term inside $\mathcal{O}(\cdot)$ using Cauchy-Schwartz inequality, we get
\begin{align}
  \hskip 2em & \hskip -2em
    \EXP\Bigg[ X_T \sum_{k=1}^{K_T} \sum_{t = t_k}^{t_{k+1} - 1} 
  \| (\theta_1 - \bar \theta_k)^\TRANS z_t \| \biggr]
  \notag \\
  &\le 
  \sqrt{ 
    \EXP\Bigg[ \sum_{k=1}^{K_T} \sum_{t = t_k}^{t_{k+1} - 1} 
  \| \Sigma_{t_k}^{-0.5} (\theta_1 - \bar \theta_k) \|^2 \Biggr]}
  \notag\\
  &\quad \times
    \sqrt{\EXP\Bigg[ \sum_{k=1}^{K_T} \sum_{t = t_k}^{t_{k+1} - 1} 
    X_T^2 \| \Sigma_{t_k}^{0.5} z_t \|^2 \Biggr]}
  \label{eq:two-terms}
\end{align}
Note that~\eqref{eq:two-terms} is slightly different than the simplification
of~\cite[Eq.~(45)]{ouyang2019posterior} using Cauchy-Schwartz inequality presented in~\cite[Eq.~(46)]{ouyang2019posterior}, which used $\Sigma_t$ in each term
in the right hand side instead of $\Sigma_{t_k}$. 

We bound each term of~\eqref{eq:two-terms} separately as follows.
\begin{lemma}\label{lem:term1}
  We have the following inequality
  \begin{align*}
    \hskip 2em & \hskip -2em
    \EXP\Bigg[ \sum_{k=1}^{K_T} \sum_{t = t_k}^{t_{k+1} - 1} 
    \| \Sigma_{t_k}^{-0.5} (\theta_1 - \bar \theta_k) \|^2 \Biggr]
    \notag\\
    &\le 
    \mathcal{O}( n(n+m) (T + \EXP[K_T]))
    \le \mathcal{O}(n(n+m)T).
  \end{align*}
\end{lemma}
See Appendix~\ref{app:term1} for a proof.
\begin{lemma}\label{lem:term2}
  We have the following inequality
  \[
    \EXP\Bigg[ \sum_{k=1}^{K_T} \sum_{t = t_k}^{t_{k+1} - 1} 
    X_T^2 \| \Sigma_{t_k}^{0.5} z_t \|^2 \Biggr]
    \le \tildeO\bigl( (n+m) \sigma_w^4 \bigr)
  \]
\end{lemma}
See Appendix~\ref{app:term2} for a proof.

We get the bound on $R_2(T)$ by substituting the result of
Lemmas~\ref{lem:term1} and~\ref{lem:term2} in~\eqref{eq:two-terms}. 

\section{Discussion and Conclusion}

In this paper, we present a variation of the TSDE algorithm of~\cite{ouyang2019posterior} and show that its Bayesian regret up to time~$T$ is bounded by $\tildeO(\sqrt{T})$ under a milder technical assumption than~\cite{ouyang2019posterior}. 
The result in~\cite{ouyang2019posterior} was derived under the assumption that there exists a $\delta < 1$ such that for any $\theta, \phi \in \Omega_1$, $\| A_\theta + B_\theta G(\phi) \| \le \delta$. For our analysis, we impose a different assumption for the closed loop gain when the system  dynamics are $\theta$ and the controller is chosen according to $\phi$.
We show that the assumption of~\cite{ouyang2019posterior} implies our assumption. Our assumption is also implied by  $\rho( A_\theta + B_\theta G(\phi) ) \le \delta$.

The key technical result in~\cite{ouyang2019posterior} as well as our paper is Lemma~\ref{lem:Xq}, which shows that for any $q \ge 1$, $\EXP[ X_T^q/ \sigma_w^q ] \le \tildeO( \log T)$. The proof argument in both~\cite{ouyang2019posterior} as well as our paper is to show that there is some constant $\alpha_0$ such that $X_T \le \sigma_w + \alpha_0 W_T$. Under the stronger assumption in~\cite{ouyang2019posterior}, one can show that for all $t$, $\| x_{t+1} \| \le \delta \| x_t \| + \| w_t \|$, which directly implies that $X_T \le \sigma_w + W_T/(1 - \delta)$. Under the weaker assumption in this paper, the argument is more subtle. The basic intuition is that in each episode, the system is asymptotically stable and, being a linear system, also exponentially stable (in the sense of Lemma~\ref{lem:bound}). So, if the episode length is sufficiently long, then we can ensure that $\| x_{t_{k+1}} \| \le \beta \| x_{t_k} \| + \bar \alpha W_T$, where $\beta < 1$ and $\bar \alpha$ is a constant. This is sufficient to ensure that $X_T \le \sigma_w + \alpha_0 W_T$ for an appropriately defined $\alpha_0$. 

The fact that each episode must be of length $T_{\min}$ implies that the second triggering condition is not triggered for the first $T_{\min}$ steps in an episode. Therefore, in this interval, the determinant of the covariance can be smaller than half of its value at the beginning of the episode. Consequently, we cannot use the same proof argument as~\cite{ouyang2019posterior} to bound $R_2(T)$ because that proof relied on the fact that for any $t \in \{t_k, \dots, t_{k+1} - 1\}$, $\det \Sigma_{t}^{-1}/ \det \Sigma_{t_k}^{-1} \le 2$. So, we provide a variation of that proof argument, where we use a coarser bound on $\det \Sigma_{t}^{-1}/ \det \Sigma_{t_k}^{-1}$ given by Lemma~\ref{lem:det-ratio}.

We conclude by observing that the milder technical assumption imposed in this paper may not be necessary. Numerical experiments indicate that the regret of the TSDE algorithm shows $\tildeO(\sqrt{T})$ behavior even when the uncertainty set $\Omega_1$ does not satisfy Assumption~\ref{ass:radius} (as was also reported in~\cite{ouyang2019posterior}). This suggests that it might be possible to further relax Assumption~\ref{ass:radius} and still establish an $\tildeO(\sqrt{T})$ regret bound.

\section*{Acknowledgement}

We would like to thank Borna Sayedana for pointing out a mistake in Lemma~9 in a previous draft of this paper. 

\appendices
\section{An auxiliary result}
\begin{lemma}\label{lem:aux}
For any $q>0$, we have 
\begin{equation}
    \EXP\Biggl[\max_{1 \le t \le T} \frac{\| w_t\|^q}{\sigma_w^q}\Biggr] \leq \mathcal{O}({(\log T)^{q/2}})
\end{equation}
\end{lemma}
\begin{proof}
 \def\1{\bar w^*_t}
 \def\0{\bar w^{(q)}_t}
 \def\2{\bar w^{(q)}_{(T)}}
  For ease of notation, define random variables $\0 = \|w_t\|^q/\sigma_w^q$ and $\1 = (\bar w^{(q)}_t)^{2/q}$. Note that $\1$ has a $\chi^2$-distribution with $n$-degrees of freedom. Therefore, $\1$ has a moment generating function $\EXP[e^{s \1}] = (1 - 2s)^{-n/2}$ for $s < 1/2$. Pick a $\lambda \in (0,1/2)$. By Chernoff bound, we have 
  \[
      \PR(\1 > z) \le \frac{\EXP[e^{\lambda \1}]}{e^{\lambda z}} =C_\lambda e^{-\lambda z},
  \]
  where $C_\lambda = (1-2\lambda)^{-n/2}$.
  Therefore, the complementary CDF of $\0$ is bounded by
  \[
      1 - F_{\0}(z) = \PR(\0 > z) = \PR(\1 > z^{q/2}) \le C_\lambda e^{-\lambda z^{q/2}}.
  \]
  
  Now, pick $\lambda' \in (0, \lambda)$ and consider an i.i.d.\@ process $\{\xi_t\}_{t \ge 1}$, where $\xi_t$ has a CDF $F_{\xi_t}(z) = 1 - e^{-\lambda' z^{q/2}}$. We let $\xi_{(T)}$ denote $\max_{1 \le t \le T} \xi_t$ and use a similar notation for $\2$. Note that $\xi_t$ has a Weibull distribution with shape $q/2$. Therefore, (see e.g., \cite[Eq. (3)]{downey1990stochastic})
  \begin{equation}\label{eq:extreme}
    \EXP[\xi_{(T)}] = \mathcal{O}( (\log T)^{q/2}).
  \end{equation}
  
  Now we present a bound on $\EXP[\2]$ in terms of $\EXP[\xi_{(T)}]$.  
  Since $0 < \lambda' < \lambda$, there exists a $z^\circ > 0$ such that for all $z > z^\circ$, 
  $F_{\0}(z) > F_{\xi_t}(z)$. 
  Thus, $\0$ is stochastically dominated by $\xi_t$ in the \emph{weak} stochastic order\footnote{A random variable $x$ is said to be dominated by a random variable $y$ in the weak stochastic order if for all increasing functions $f$ supported sufficiently away from $0$, $\EXP[f(x)] \le \EXP[f(y)]$.}, as defined in~\cite{downey1990stochastic}. Therefore, by \cite[Theorem~2.1]{downey1990stochastic}, there exists a constant $c > 0$ such that $\0$ is stochastically dominated by $c \xi_t$ in the convex order.\footnote{A random variable $x$ is said to be dominated by a random variable $y$ in the convex order if for all increasing and convex functions $f$, $\EXP[f(x)] \le \EXP[f(y)]$.} 
  Consequently, by \cite[Theorem~3.1(1)]{downey1990stochastic} (or \cite[Theorem 2.2(1)]{downey1992orderings}), we have that $\EXP[\2] \le c\EXP[\xi_{(T)}]$. Substituting~\eqref{eq:extreme} establishes the result of the Lemma.
\end{proof}

\section{Proof of Lemma~\ref{lem:Xq}}
\label{app:Xq}
For the ease of notation, let $\bar
\alpha = \alpha/(1 - \delta)$, and $\beta
= \alpha \delta^{T_{\min} +1}$. In addition, define $W_T = \max_{1 \le t
\le T} \| w_t\|$, $\bar X_k = \max_{t_k < t \le t_{k+1}} \| x_t \|$, $Y_k = \|
x_{t_k} \|$, and $H_k = A + B G(\bar \theta_k)$ where $A$ and $B$ are the true parameters.

From the system dynamics under the TSDE algorithm, we know that for any
time $t \in \{ t_k + 1, \dots, t_{k+1} \}$, we have
\[ 
x_t = H_k^{t -t_k} x_{t_k} + \sum_{j = t_k}^{t-1} H_k^{t-1-j}w_j.
\]

Thus, from triangle inequality and Assumption~\ref{ass:type}, we get
\begin{align}
 \|x_t\| &\le \alpha \delta^{t- t_k} Y_k +  \biggl[\sum_{j = t_k}^{t-1}\alpha \delta^{t- 1-j}\biggr]W_T
  \notag \\
  &\le \alpha \delta^{t- t_k} Y_k + 
  \underbrace{\biggl[ \frac{\alpha }{1 - \delta} \biggr]}_{\eqqcolon \bar \alpha} W_T.
  \label{eq:x-norm}
\end{align}

Now at time $t = t_{k+1}$, we have 
\begin{align}
  Y_{k+1} &= \| x_{t_{k+1}}\| \le 
  \alpha \delta^{T_k} Y_k + \bar \alpha W_T.
  \notag\\
  &\le \beta Y_k + \bar \alpha W_T
  \label{eq:Y}
\end{align}
where the second inequality follows from~\eqref{eq:min-length}, which implies
$\alpha \delta^{T_k} \le \alpha \delta^{T_{\min}+1} \eqqcolon
\beta$. From Lemma~\ref{lem:t-min}, $\beta<1$. Recursively expanding~\eqref{eq:Y}, we get
\begin{align}
  Y_k &\le \bar \alpha W_T + \beta \bar \alpha W_T + \dots + \beta^{k-2} \bar
  \alpha W_T 
  \notag \\
  &\le \frac{\bar \alpha}{1-\beta} W_T
  \eqqcolon \bar \beta W_T.
  \label{eq:Y-bound}
\end{align}

Substituting~\eqref{eq:Y-bound} is~\eqref{eq:x-norm}, we get that for any 
$t \in \{ {t_k + 1}, \allowbreak \dots, t_{k+1} \}$, we have
\[
  \| x_t || \le \alpha \delta^{t-t_k} \bar \beta W_T + \bar \alpha W_T
  \le \underbrace{[ \alpha \bar \beta + \bar \alpha ]}_{\eqqcolon \alpha_0} W_T
\]
where in the last inequality, we have used the fact that $\delta \in (0,
1)$. Thus, for any episode~$k$, we have 
\begin{equation*}
  \bar X_k = \max_{t_{k} < t \le t_{k+1}} \| x_t \| \le
  \alpha_0 W_T.
\end{equation*}
Hence,
\[
  X_T  \le \sigma_w + \max\{ \bar X_1, \dots, \bar X_{K_T} \} \le \sigma_w + \alpha_0 W_T.
\]

Therefore, for any $q \ge 1$, we have
\begin{equation}\label{eq:Xq-bound}
  \EXP[ X_T^q ] \le \sum_{p=0}^q \binom{q}{p}\sigma_w^{q - p}\alpha_0^p \EXP[ W_T^p ]
\end{equation}
From Lemma~\ref{lem:aux}, we have that
\[
  \sigma_w^{q - p}\EXP[ W_T^p ] =
  \sigma_w^{q - p}\EXP\Bigl[ \max_{1 \le t \le T} \| w_t\|^p \Bigr]
  \le \sigma_w^q \mathcal{O}( {(\log T)}^{p/2}).
\]
Substituting this is~\eqref{eq:Xq-bound}, we obtain the result of the lemma.

\section{Proof of Lemma~\ref{lem:log-bound}} \label{app:log-bound}
Since log is an increasing function, $\log x \le \log \max(e, x)$ for any $x > 0$. Therefore, 
\begin{align}
  \EXP\Big[ &\frac{X_T^q}{\sigma_w^q} \log
     \Big(\frac{X_T^2}{\sigma_w^2}\Big) \Big] 
    \le 
  \EXP\Big[ \frac{X_T^q}{\sigma_w^q} \log \max(e, X_T^2/\sigma_w^2) \Big]
  \notag \\
  &\le \sqrt{ \EXP\Big[ \frac{X_T^{2q}}{\sigma_w^{2q}} \Big]\; \EXP\Big[ \Big( \log \max(e, X_T^2/\sigma_w^2) \Big)^2 \Big]}
  \label{eq:log-bound-1}
\end{align}
where the last inequality follows from Cauchy-Schwartz inequality.  Since $(\log x)^2$ is concave for $x \ge e$, we can use Jensen's inequality to write
\begin{align}
  \EXP\Big[ \Big( \log \max(e, X_T^2/\sigma_w^2) \Big)^2 \Big] 
  &\le
  \bigl(\log( \EXP[ \max(e, X_T^2/\sigma_w^2) ])\bigr)^2 
  \notag \\
  &\le
  \bigl(
    \log( e + \EXP[ X_T^2 / \sigma_w^2])
  \bigr)^2
  \notag \\
  &\stackrel{(a)}\le
  \bigl(
    \log(e + \mathcal{O}(\log T))
  \bigr)^2
  \notag \\
  &\le \tildeO(1)
  \label{eq:log-bound-2}
\end{align}
where $(a)$ uses Lemma~\ref{lem:Xq}. Substituting~\eqref{eq:log-bound-2} in~\eqref{eq:log-bound-1} and using Lemma~\ref{lem:Xq} for bounding $\EXP[ X_T^{2q} / \sigma_w^{2q}]$, we get
\begin{align*}
  \EXP\Big[ \frac{X_T^q}{\sigma_w^q} \log
     \Big(\frac{X_T^2}{\sigma_w^2}\Big) \Big] 
  &\le \sqrt{ \EXP\Big[ \frac{X_T^{2q}}{\sigma_w^{2q}} \Big]\; \tildeO(1)}
  \le \tildeO(1).
\end{align*}

\section{Proof of Lemma~\ref{lem:Kt}}
\label{app:Kt}

The high-level idea of the proof is same as that
of~\cite[Lemma~3]{ouyang2019posterior}. Define macro episodes with start times
  $t_{n_i}$, $i \in \integers_{> 0}$, where ${n_1} = 1$ and for $i \ge
1$, 
\[
  {n_{i + 1}} = 
     \min \left\{ k > {n_i} 
      \,\middle|\, 
      \det \Sigma_{t_k} < \tfrac 12 \det \Sigma_{t_{k-1}} 
  \right\}.
\]

Thus, a new macro-episode starts whenever an episode ends due to the second
stopping criterion. Let $M$ denote the number of macro-episodes until
time~$T$ and define $n_{M+1} = K_T +1$. Let $\bar T_i$ denote the length of
the $i$-th macro-episode. Within a macro-episode, all but the last episode
must be triggered by the first stopping criterion. Thus, for $k \in \{ n_i,
n_i + 1, \dots, n_{i+1} - 2 \}$, 
\[
  T_k = \max\{ T_{k-1} +1, T_{\min}+1 \} = T_{k-1} + 1
\]
where the last equality follows from~\eqref{eq:min-length}. Hence, by
following exactly the same argument as~\cite{ouyang2019posterior}, we have
\[
  n_{i+1} - n_i \le \sqrt{2 \bar T_i}
\]
and therefore following~\cite[Eq.~(40)]{ouyang2019posterior}, we have
\begin{equation}\label{eq:Kt-bound}
  K_T \le \sqrt{2 MT}
\end{equation}
which is same as~\cite[Eq.~(41)]{ouyang2019posterior}.

Now, observe that
\begin{align}
  \det \Sigma_T^{-1}
  & \stackrel{(a)}\ge
  \det \Sigma_{t_{n_M}}^{-1}
  \stackrel{(b)}\ge
  2 \det \Sigma_{t_{n_{M - 1}}}^{-1}
  \notag \\
  & \ge \cdots \ge 2^{M-1} \det \Sigma_1^{-1},
  \label{eq:det-Sigma}
\end{align}
where $(a)$ follows because $\{ \det \Sigma_t^{-1} \}_{t \ge 1}$ is a
non-decreasing sequence (because $\Sigma_1^{-1} \le \Sigma_2^{-1} \ldots$) and~$(b)$ and subsequent
inequalities follow from the definition of the macro
episode and the second triggering condition.

Then following the same idea as the rest of the proof
in~\cite{ouyang2019posterior}, we get
\begin{equation}\label{eq:M-bound}
  M \le \mathcal{O}( (n+m) \log (T X_T^2 / \sigma_w^2)).
\end{equation}
Substituting~\eqref{eq:M-bound} in~\eqref{eq:Kt-bound}, we obtain the result
of the lemma.

\section{Proof of Lemma~\ref{lem:term1}} \label{app:term1}
Observe that the summand is constant for each episode. Therefore,
\begin{align}
  \hskip 1em & \hskip -1em
  \EXP\biggl[\sum_{k=1}^{K_T} \sum_{t=t_k}^{t_{k+1}-1} \bigl[ 
    \| \Sigma_{t_k}^{-0.5} (\theta_1 - \bar \theta_k)^\TRANS \|^2
\bigr] \biggr]
  \notag \\
  &= \EXP\biggl[\sum_{k=1}^{K_T}  \bigl[ 
    T_{k}
    \| \Sigma_{t_k}^{-0.5} (\theta_1 - \bar \theta_k)^\TRANS \|^2
  \bigr]\biggr]
  \displaybreak[1]
  \notag \\
  &\stackrel{(a)}\le \EXP\biggl[\sum_{k=1}^{K_T}  \bigl[ 
  (T_{k-1}+1)
    \| \Sigma_{t_k}^{-0.5} (\theta_1 - \bar \theta_k)^\TRANS \|^2
  \bigr]\biggr]
  \displaybreak[1]
  \notag \\
  &= \sum_{k=1}^\infty \EXP \bigl[ \IND_{\{t_k \le T\}} 
    (T_{k-1}+1)
    \| \Sigma_{t_k}^{-0.5} (\theta_1 - \bar \theta_k)^\TRANS \|^2
  \bigr]
  \displaybreak[1]
  \notag\\
  &= \sum_{k=1}^\infty \EXP \Bigl[ \EXP \bigl[ \IND_{\{t_k \le T\}} 
      (T_{k-1}+1)
    \| \Sigma_{t_k}^{-0.5} (\theta_1 - \bar \theta_k)^\TRANS \|^2
   \bigm| h_{t_k} \bigr]\Bigr]
  \displaybreak[1]
   \notag \\
   &\stackrel{(b)}=
   \sum_{k=1}^\infty \EXP \Bigl[ \IND_{\{t_k \le T\}} 
     (T_{k-1}+1)\EXP \bigl[ 
    \| \Sigma_{t_k}^{-0.5} (\theta_1 - \bar \theta_k)^\TRANS \|^2
   \bigm| h_{t_k} \bigr]\Bigr]
  \displaybreak[1]
   \notag \\
   &\stackrel{(c)}\le
   \sum_{k=1}^\infty \EXP \bigl[ \IND_{\{t_k \le T\}} 
   (T_{k-1}+1) 2(n+m)n \bigr]
   \notag \\
   &\le
   2(n + m)n (T  + \EXP[K_T]),
   \label{eq:pf-term1}
\end{align}
where $(a)$ follows from~\eqref{eq:min-length}, $(b)$ follows from the fact that
$\IND_{\{t_k < T\}}(T_{k-1}+1)$ is $\sigma(h_{t_k})$ measurable, and
$(c)$ hold because conditioned on $h_{t_k}$ each column of
$\| \Sigma_{t_k}^{-0.5} (\theta_1 - \bar \theta_k)^\TRANS \|^2$ is the
difference of two i.i.d.\@ vectors $\sim \mathcal{N}(0,I)$. 

Eq.~\eqref{eq:pf-term1} proves the first part of the Lemma. The second part
follows from the fact that $K_T \le T$.

\section{Proof of Lemma~\ref{lem:term2}}\label{app:term2}
For any $s < t$. 
Eq.~\eqref{eq:Sigma} implies that $\Sigma_s^{-1}
\preceq \Sigma_t^{-1}$ and consequently 
$\Sigma_t^{-1}$ is positive definite. Therefore, we have the following:
\begin{lemma}\label{lem:eval}
  Let $\lambda_{\min}$ be the smallest eigenvalue of $\Sigma_1^{-1}$. Then, each
  eigenvalue of $\Sigma_t^{-1}$ is no less than $\lambda_{\min}$. Therefore,
  each eigenvalue of $\Sigma_t$ is no more than $1/\lambda_{\min}$.
\end{lemma}
An immediate implication of Lemma~\ref{lem:eval} is the following: For any $t$
and $s$, 
\begin{equation}\label{eq:quad-form-bound}
  z_t^\TRANS \Sigma_s z_t \le \frac{1}{\lambda_{\min}} \| z_t \|^2
  \le \frac{1}{\lambda_{\min}} M^2_G X_T^2,
\end{equation}
where $M_G = \sup_{\theta \in \Omega_1} \|[I, G(\theta)^\TRANS]^\TRANS \|$. 

For any $s < t$, $\Sigma_s^{-1} \preceq \Sigma_t^{-1}$ implies that $\Sigma_s
\succeq \Sigma_t$. Therefore, from \cite[Lemma~11]{abbasi2011regret}, we get
that for any $V \neq 0$ (of appropriate dimensions), 
\begin{equation}\label{eq:det-bound}
  \frac{ \| V^\TRANS \Sigma_s V\| }
       { \| V^\TRANS \Sigma_t V \| }
  \le \frac{ \det \Sigma_s }{ \det \Sigma_t }
  = \frac{ \det \Sigma_t^{-1} }{ \det \Sigma_s^{-1}}.
\end{equation}

Eq.~\eqref{eq:det-bound} implies that for any $t \in \{t_k, \dots, t_{k+1} -
1\}$, we have
\begin{equation}
    \| \Sigma_{t_k}^{0.5} z_t \|^2 
    =
    z_t^\TRANS \Sigma_{t_k} z_t 
    \le 
    \frac{\det \Sigma_t^{-1}}{\det \Sigma_{t_k}^{-1} }
    z_t^\TRANS \Sigma_{t} z_t 
    \label{eq:term2-bound}
\end{equation}

For the ease of notation, let $\tau_k = t_k + T_{\min}$. Then we have the
following bound on $\det \Sigma_{t}^{-1}/\det \Sigma_{t_k}^{-1}$.
\begin{lemma}\label{lem:det-ratio}

  The following inequalities hold:
  \begin{enumerate}
    \item For  $t \in \{ t_k, \dots, \tau_k\}$, we have
      \[
        \frac{\det \Sigma_t^{-1}}{\det \Sigma_{t_k}^{-1} }
        \le
        \biggl( 1 + \frac{1}{\lambda_{\min} \sigma_w^2} M^2_G X_T^2 \biggr)^{T_{\min}}.
      \]
    \item For $t \in \{\tau_k + 1, \dots, t_{k+1} -1 \}$, we have
      \[
        \frac{\det \Sigma_t^{-1}}{\det \Sigma_{t_k}^{-1} }
        \le
        2.
      \]
  \end{enumerate}
  
Consequently,  for  all $t \in \{ t_k, \dots, t_{k+1}-1\}$, we have
    \begin{equation}\label{eq:det-bound-1}
        \frac{\det \Sigma_t^{-1}}{\det \Sigma_{t_k}^{-1} }
        \le
        \biggl(2 + \frac{M^2_G X_T^2 }{\lambda_{\min} \sigma_w^2} \biggr)^{T_{\min} \vee 1}.
    \end{equation}
\end{lemma}
\begin{proof}
  The second relationship follows from the second stopping criterion. We now
  prove the first relationship. Eq.~\eqref{eq:Sigma} implies that
  \[
     \Sigma_{t+1}^{-1} = \Sigma_t^{-1}
     \biggl( I + \frac{1}{\sigma_w^2} \Sigma_t z_t z_t^\TRANS \biggr).
  \]
  Therefore, 
  \begin{align}
    \frac{\det \Sigma_{t+1}^{-1}}{\det \Sigma_{t}^{-1} }
    &= \det \biggl( I + \frac{1}{\sigma_w^2} \Sigma_t z_t z_t^\TRANS \biggr)
    = 1 + \frac{1}{\sigma_w^2} z_t^\TRANS \Sigma_t z_t
    \notag \\
    &\le 1 + \frac{1}{\lambda_{\min} \sigma_w^2 } M^2_G X_T^2,
    \label{eq:det-ratio1}
  \end{align}
  where the last inequality follows from~\eqref{eq:quad-form-bound}.
  Thus, for any $t \in \{t_k, \dots, \tau_k\}$, we have
  \begin{align}
    \frac{\det \Sigma_t^{-1}}{\det \Sigma_{t_k}^{-1} }
   & \le \biggl( 1 + \frac{1}{\lambda_{\min} \sigma_w^2} M^2_G X_T^2 \biggr)^{t - t_k} \notag \\
    &\le \biggl( 1 + \frac{1}{\lambda_{\min} \sigma_w^2} M^2_G X_T^2 \biggr)^{T_{\min}}
  \end{align}
  where the first inequality follows by repeatedly applying~\eqref{eq:det-ratio1}
  as a telescopic product.

  Let $\bar M = M_G^2 X_T^2/\lambda_{\min} \sigma_w^2$. Then, \eqref{eq:det-bound-1} follows by observing that $(1 + \bar M)^{T_{\min}} \le (2 + \bar M)^{T_{\min} \vee 1}$ and $2 < (2 + \bar M)^{T_{\min} \vee 1}$.
\end{proof}

Using Lemma~\ref{lem:det-ratio} and \eqref{eq:term2-bound}, we get
\begin{align}
  \sum_{t={t_k}}^{t_{k+1} - 1} \| \Sigma_{t_k}^{0.5} z_t\|^2 
  &\le
\sum_{t={t_k}}^{t_{k+1} - 1}  \frac{\det \Sigma_t^{-1}}{\det \Sigma_{t_k}^{-1} } 
z_t^\TRANS \Sigma_{t} z_t 
  \notag \\
  &\le  \biggl(2 + \frac{M^2_G X_T^2 }{\lambda_{\min} \sigma_w^2}  \biggr)^{T_{\min} \vee 1}
  \sum_{t={t_k}}^{t_{k+1} - 1} 
z_t^\TRANS \Sigma_{t} z_t 
\end{align}
where the first inequality follows from~\eqref{eq:term2-bound} and the second inequality follows
from Lemma~\ref{lem:det-ratio}.
Therefore,
\begin{align}
  \hskip 2em & \hskip -2em
  \sum_{k=1}^{K_T} \sum_{t={t_k}}^{t_{k+1} - 1} X_T^2 \| \Sigma_{t_k}^{0.5} z_t\|^2 
  \notag \\
  &\le
  \biggl(2 + \frac{M^2_G X_T^2 }{\lambda_{\min} \sigma_w^2} \biggr)^{T_{\min} \vee 1}  X_T^2 
  \sum_{t=1}^{T}  z_t^\TRANS \Sigma_{t} z_t 
  \label{eq:R2-pf-split}
\end{align}

From~\eqref{eq:quad-form-bound} for $s=t$, we get that
\begin{equation}\label{eq:max-min-bound}
  z_t^\TRANS \Sigma_t z_t \le
  \max\biggl(\sigma_w^2, \frac{M_G^2 X_T^2}{\lambda_{\min}}\biggr)
  \min\biggl(1, \frac{z_t^\TRANS \Sigma_t z_t}{\sigma_w^2}\biggr).
\end{equation}

Hence
\begin{equation} \label{eq:term2:2}
  \sum_{t=1}^T z_t^\TRANS \Sigma_t z_t \le
  \biggl(\sigma_w^2  + \frac{M_G^2 X_T^2}{\lambda_{\min}}\biggr)
  \sum_{t=1}^T 
  \min\biggl(1, \frac{z_t^\TRANS \Sigma_t z_t}{\sigma_w^2}\biggr)
\end{equation}
Using~\eqref{eq:Sigma} and the intermediate step of the proof
of~\cite[Lemma 6]{abbasiyadkori2014bayesian}, we have
\begin{align}
  \hskip 1em & \hskip -1em 
  \sum_{t=1}^T
  \min\biggl(1, \frac{z_t^\TRANS \Sigma_t z_t}{\sigma_w^2} \biggr) 
  =  \sum_{t=1}^T \min\biggl(1, \biggl\lVert\frac{\Sigma^{0.5}_tz_tz_t^\TRANS \Sigma^{0.5}_t}{\sigma_w^2}\biggr\rVert \biggr) \notag \\
  &\le 2(n + m) \log \Bigg( \frac{ \TR( \Sigma_{T+1}^{-1 }) }{ (n+m) } \Biggr) -
  \log\det \Sigma^{-1}_1.
  \label{eq:R2-pf-AY}
\end{align}
Now, from~\eqref{eq:Sigma}, we get that
\begin{align}
  \TR(\Sigma_{T+1}^{-1}) &= \TR(\Sigma_1^{-1}) + 
  \sum_{t=1}^{T} \frac{1}{\sigma_w^2} \TR(z_t z_t^\TRANS)
  \notag\\
  &\le \TR(\Sigma_1^{-1}) + 
   \frac{T}{\sigma_w^2} M_G^2 X_T^2,
  \label{eq:tr-bound}
\end{align}
where the last inequality uses the fact that $\TR(z_t z_t^\TRANS)=
\TR(z_t^\TRANS z_t) = \| z_t\|^2 \le M_G^2 X_T^2$.
Combining~\eqref{eq:term2:2} with~\eqref{eq:R2-pf-AY} and~\eqref{eq:tr-bound},  we get
\begin{equation}
  \sum_{t=1}^T z_t^\TRANS \Sigma_t z_t \le
  \mathcal{O}\bigl(
    (n+m)(\sigma_w^2 + X_T^2) \log( T X_T^2 / \sigma_w^2) 
  \bigr).
  \label{eq:zz-bound}
\end{equation}
Therefore, we can bound the expectation of the right hand side of~\eqref{eq:R2-pf-split} as
\begin{align}
  \hskip 2em & \hskip -2em
    \EXP\biggl[ 
 \biggl(2 + \frac{M^2_G X_T^2 }{\lambda_{\min} \sigma_w^2} \biggr)^{T_{\min} \vee 1} X_T^2
  \sum_{t=1}^{T}  z_t^\TRANS \Sigma_{t} z_t  
  \biggr]
    \notag\\
    &\le \mathcal{O}\bigl( 
       \sigma_w^4 (n+m)
        \EXP[ F(X_T) ]
    \bigr)
    \notag\\
    &\le \tildeO(\sigma_w^4(n+m)),
    \label{eq:term2:4}
\end{align}
where the first inequality follows from~\eqref{eq:zz-bound} with $F(X_T) =  \big(2 + \frac{M^2_G X_T^2 }{\lambda_{\min} \sigma_w^2} \big)^{T_{\min} \vee 1} (\frac{X_T^2}{\sigma_w^2} + \frac{X_T^4}{\sigma_w^4}) \log( T X_T^2/ \sigma_w^2) $, 
and the last inequality follows from Lemma~\ref{lem:log-bound} by noting that $F(X_T)$ is a polynomial of $X_T/\sigma_w$ multiplied by a poly-log term.

The result follows from \eqref{eq:R2-pf-split} and \eqref{eq:term2:4}.


\begin{thebibliography}{10}

\bibitem{ouyang2019posterior}
Y.~Ouyang, M.~Gagrani, and R.~Jain, ``Posterior sampling-based reinforcement
  learning for control of unknown linear systems,'' {\em {IEEE} Trans. Autom.
  Control}, vol.~65, no.~8, pp.~3600--3607, 2019.

\bibitem{astrom1994adaptive}
K.~J. Astrom and B.~Wittenmark, {\em Adaptive Control}.
\newblock Addison-Wesley Longman Publishing Co., Inc., 1994.

\bibitem{Caines:1988}
P.~E. Caines, {\em Linear Stochastic Systems}.
\newblock John Wiley, 1988.
\newblock Republished by {\it Society of Industrial Applied Mathematics,} 2018.

\bibitem{bradtke1993reinforcement}
S.~J. Bradtke, ``Reinforcement learning applied to linear quadratic
  regulation,'' in {\em Neural Information Processing Systems}, pp.~295--302,
  1993.

\bibitem{bradtke1994adaptive}
S.~J. Bradtke, B.~E. Ydstie, and A.~G. Barto, ``Adaptive linear quadratic
  control using policy iteration,'' in {\em Proceedings of American Control
  Conference}, vol.~3, pp.~3475--3479, IEEE, 1994.

\bibitem{campi1998adaptive}
M.~C. Campi and P.~Kumar, ``Adaptive linear quadratic {G}aussian control: the
  cost-biased approach revisited,'' {\em SIAM Journal on Control and
  Optimization}, vol.~36, no.~6, pp.~1890--1907, 1998.

\bibitem{abbasi2011regret}
Y.~Abbasi-Yadkori and C.~Szepesv{\'a}ri, ``Regret bounds for the adaptive
  control of linear quadratic systems,'' in {\em Annual Conference on Learning
  Theory}, pp.~1--26, 2011.

\bibitem{cohen2019learning}
A.~Cohen, T.~Koren, and Y.~Mansour, ``Learning linear-quadratic regulators
  efficiently with only $\sqrt{T} $ regret,'' in {\em International Conference
  on Machine Learning}, pp.~1300--1309, PMLR, 2019.

\bibitem{abeille2020efficient}
M.~Abeille and A.~Lazaric, ``Efficient optimistic exploration in
  linear-quadratic regulators via lagrangian relaxation,'' in {\em
  International Conference on Machine Learning}, pp.~23--31, PMLR, 2020.

\bibitem{dean2018regret}
S.~Dean, H.~Mania, N.~Matni, B.~Recht, and S.~Tu, ``Regret bounds for robust
  adaptive control of the linear quadratic regulator,'' in {\em Neural
  Information Processing Systems}, pp.~4192--4201, 2018.

\bibitem{mania2019certainty}
H.~Mania, S.~Tu, and B.~Recht, ``Certainty equivalent control of {LQR} is
  efficient.'' arXiv:1902.07826, 2019.

\bibitem{faradonbeh2020input}
M.~K.~S. Faradonbeh, A.~Tewari, and G.~Michailidis, ``Input perturbations for
  adaptive control and learning,'' {\em Automatica}, vol.~117, p.~108950, 2020.

\bibitem{simchowitz2020naive}
M.~Simchowitz and D.~Foster, ``Naive exploration is optimal for online lqr,''
  in {\em International Conference on Machine Learning}, pp.~8937--8948, PMLR,
  2020.

\bibitem{faradonbeh2020adaptive}
M.~K.~S. Faradonbeh, A.~Tewari, and G.~Michailidis, ``On adaptive
  linear–quadratic regulators,'' {\em Automatica}, vol.~117, p.~108982, July
  2020.

\bibitem{ouyang2017control}
Y.~Ouyang, M.~Gagrani, and R.~Jain, ``Control of unknown linear systems with
  thompson sampling,'' in {\em Allerton Conference on Communication, Control,
  and Computing}, pp.~1198--1205, 2017.

\bibitem{abeille2018improved}
M.~Abeille and A.~Lazaric, ``Improved regret bounds for thompson sampling in
  linear quadratic control problems,'' in {\em International Conference on
  Machine Learning}, pp.~1--9, 2018.

\bibitem{gagrani2021thompson}
M.~Gagrani, S.~Sudhakara, A.~Mahajan, A.~Nayyar, and Y.~Ouyang, ``Thompson
  sampling for linear quadratic mean-field teams,'' in {\em 2021 60th IEEE
  Conference on Decision and Control (CDC)}, pp.~720--727, IEEE, 2021.

\bibitem{Astrom1970}
K.~J. Astrom, {\em Introduction to stochastic control theory}.
\newblock Academic Press New York, 1970.

\bibitem{sternby1977consistency}
J.~Sternby, ``On consistency for the method of least squares using martingale
  theory,'' {\em {IEEE} Trans. Autom. Control}, vol.~22, no.~3, pp.~346--352,
  1977.

\bibitem{downey1990stochastic}
P.~J. Downey and R.~S. Maier, ``Stochastic orderings and the growth of expected
  extremes,'' tech. rep., Tech. Report 90-9, Department of Computer Science,
  University of Arizona, 1990.

\bibitem{downey1992orderings}
P.~J. Downey and R.~S. Maier, ``Orderings arising from expected extremes, with
  an application,'' in {\em Institute of {Mathematical} {Statistics} {Lecture}
  {Notes} - {Monograph} {Series}}, pp.~66--75, Hayward, CA: Institute of
  Mathematical Statistics, 1992.

\bibitem{abbasiyadkori2014bayesian}
Y.~Abbasi-Yadkori and C.~Szepesvari, ``Bayesian optimal control of smoothly
  parameterized systems: The lazy posterior sampling algorithm.'' arXiv
  preprint arXiv:1406.3926, 2014.

\end{thebibliography}
\end{document}